\documentclass[a4paper,13pt]{article}
\usepackage{amsmath,amsfonts,amssymb}
\usepackage{setspace}
\usepackage[utf8]{inputenc}
\usepackage{multirow}
\usepackage{listings}
\usepackage[T1]{fontenc}
\usepackage{times}
\usepackage{multicol,amssymb,amsthm,amsmath,graphicx,geometry,xcolor}
\usepackage{hyperref}
\usepackage{float}
\usepackage{mathrsfs}
\newtheorem{proposition}{Proposition}

\newtheorem{lemma}{Lemma}
\newtheorem{remarque}{Remark}

\title{American Exchange options driven by a Lévy process}

\author{
 Marah Zakaria \\
  \texttt{marahzakaria1@gmail.com} \\
}

\begin{document}
\maketitle

\setstretch{1.125}

\begin{abstract}
We consider the problem of pricing American Exchange options driven by a L\'evy process.  We  study the properties of American Exchange options, we represented it as the sum of the price of the corresponding European exchange option price
and an early exercise premium. Secondly we show some properties of the free boundary and give an approximative formula of an American Exchange option.
\end{abstract}

\section{Introduction}
\subsection{Overview}
An exchange option is a contract that grants the holder the
right, but not the obligation, to exchange one risky asset for
another. It is predominantly used in foreign exchange, fixed‐income, and equity markets. In practice, many financial assets and real investment opportunities can be analyzed as American exchange options. In the Black and Scholes framework,
the price of the European exchange option is given by the celebrated (Margrabe \cite{Margrabe}) formula. Because of the limitations of the Black–Scholes framework, alternative asset price models have been proposed to provide more accurate characterizations
of asset returns. Some examples of these alternative
models are jump-diffusion models ( Pham \cite{Pham}, Lamberton \cite{Lamberton}). Under two
geometric Brownian motion processes, Broadie and Detemple \cite{Broadie} presented integral equations for early exercise
boundary and option prices for finite‐lived American spread and exchange options. Cheang and Chiarella \cite{Cheang}
presented a probabilistic representation for the American style exchange option under jump‐diffusion dynamics.\\
In this paper, we provide an extension to the results of
Margrabe \cite{Margrabe}, Cheang and Chiarella  \cite{Cheang}and Guanghua Lian et al \cite{Lian}to consider
the case where, asset prices
are driven by Lévy process. To facilitate the analysis we employs the
change-of-numéraire technique to obtain a representation that
is similar to the classical Margrabe  \cite{Margrabe} formula. Subsequently we found a  representation of European exchange option
prices in terms of the characteristic function. We also demonstrated that the American exchange
option price can also be represented as the sum of the
price of the corresponding European exchange option price
and an early exercise premium, similar to the findings of Cheang
and Chiarella \cite{Cheang} in Jump diffusion case, however Cheang and Chiarella did not show the regularity of American option to justify the use of Ito lemma. Unlike the above authors we were also able to show different properties of the free boundary thanks to the dimension reduction. Finally we give an approximate formula of an American exchange option. 
The paper is organized as follows. In subection. 1.2, we recall some basic facts about
Exchange option. In Sect. 2, we studie the American Exchange option. We obtain a decomposition of the American option value as the
sum of its corresponding European Exchange price and the early exercise premium. The ramaining parts are devoted to the properties of the exercise boundary. We first establish the continuity of the free boundary, then we study the limit of the critical price at maturity. We also provide an approximate formula for an
American exchange option, where
the dynamics of the underlying assets are driven by a L\'evy processes.
\subsection{Exchange option driven by L\'evy process}
Let X be a semimartingale on the stochastic basis $(\Omega,(\mathcal{F}_t)_{t\in \mathbb{R}^+}, \mathbb{P})$, with values in $\mathbb{P}$ and $X_0 = 0$. Suppose the stock price has the following dynamics:
\begin{align}
&dS_{i,t}=S_{i,t-} dX_{i,t},\\
&dX_{i,t}=(r-q_i)dt+\sigma dW_{i,t} + dZ_t^i\\
&Z^i_t=\int_0^t \int_{\mathbb{R}^2} (e^{y_i}-1)(J(ds,dy)-\nu (ds,dy))\\
&d<W^1,W^2>_t=\rho dt.
\end{align}
Where $W_{1,t}$ and $W_{2,t}$ are components of a bivariate correlated
Brownian motion process which is adapted to the filtration, where $d<W_1, dW_2>_t=\rho$,
and $\rho$ is the instantaneous correlation between the two Brownian motion components. 
The component $J(ds,dz)$ is a Poisson random measure with intensity measure  $\nu(dz)$. The measure $\nu$ is a positive Radon measure, called the Lévy measure of $L$, and it satisfies
\begin{align}\label{condition}
\int_{\mathbb{R}^2} e^{<u,y>} -1 \hspace{0.1cm}\nu (dy)<\infty,  \hspace{1cm} \forall u\in \mathbb{R}^2.
\end{align}
$J(ds,dy)-\nu (ds,dy)$ is the compensated Poisson random measure that corresponds to $J(ds,dy)$. \\
We will from now on assume $\mathbb{P}$ to be a risk-neutral measure and the interest
rate to be a constant r and a constant positive dividend
yield of $q_1$ and $q_2$, respectively, per annum.\\
We assume that the that the Lévy process is independent of the Brownian motions and of each
other. Henceforth, we assume that ${\mathcal{F}_t}$ is the natural filtration generated by the Brownian motions and the Lévy process.
Note that, in this framework, we have to consider payoff functions   which depend on both the time and the space variables. For example, in the case of a standard European exchange
option, the prise is $c (S_{1,t}, S_{2,t}, t, T)=\mathbb{E}^\mathbb{P}[(S_{1,t}-KS_{2,t})^+|\mathcal{F}_t]$. An American exchange option gives its owner the right to exchange one asset for another at any time prior to expiration. 
\begin{align*}
C(S_{1,t}, S_{2,t},t,T)=\sup_{\theta \in \mathcal{T}_{t,T}} \mathbb{E}^\mathbb{P}[e^{-r(\theta-t)}(KS_{1,t}-S_{2,t})^+|\mathcal{F}_t]
\end{align*}
where  $\mathcal{T}_{t,T}$ is the set of all stopping times with values in $[t, T]$. We define $\mathbb{Q}$  an adequate probability measure such that
\begin{align*}
\frac{d\mathbb{Q}}{d\mathbb{P}}=\frac{S_{2,t}e^{(q_2-r)t}}{S_{2,0} } 
\end{align*}
Under the measure $\mathbb{Q}$ we had that $R_t=\frac{e^{q_1t}S_{1,t}}{e^{q_2t}S_{2,t}}$ is a local martingale. As a conclusion the dynamics of the process $R_t$ is:
\begin{align*}
dR_t&=R_{t-}\big ( \sigma dW_t^{\mathbb{Q}}+ \int_{\mathbb{R}^2} (e^{y_1-y_2}-1)( J(dy,ds) - \tilde{\nu} (dy)ds\big),
\end{align*}
where $\sigma dW_t^{\mathbb{Q}}=\sigma_1 dW_{1,t}^{\mathbb{Q}} - \sigma_2 dW_{2,t}^{\mathbb{Q}}$ and $\tilde{\nu}=e^{y_2} \nu$.\\
Write $X_t = log (R_t )$ we obtain:
\begin{align*}
dX_t&=-\Big( \int_{\mathbb{R}^2} e^{y_1-y_2}-1 \tilde{\nu}(dy)-\int_{|y|<1} (y_1-y_2) \nu(ds,dy)+\frac{1}{2}\sigma^2\Big) dt\\
&+\sigma dW_t^\mathbb{Q}+\int_{|y|<1} (y_1-y_2) \tilde{J}(ds,dy)+\int_{|y|>1} (y_1-y_2) J(ds,dy),
\end{align*}
 Which is a L\'evy process with the characteristic exponent under $\mathbb{Q}$
\begin{align*}
f(u,X_t,t)=\mathbb{E}[e^{iuX_t}]=exp\Big (-t\Big [ \int_{\mathbb{R}^2} (e^{y_1-y_2}-1) \tilde{\nu}(dy)+\frac{1}{2}\sigma^2 \Big ]iu-\frac{1}{2}\sigma^2u^2t+ t\int_{\mathbb{R}^2} (e^{u(y_1-y_2)}-1) \tilde{\nu}(dy)\Big) 
\end{align*}
With the derived analytical‐form characteristic function, we can solve the pricing of a European exchange option as below
\begin{footnotesize}
\begin{align}\label{eq:5}
c (S_{1,t}, S_{2,t}, t, T)&=\mathbb{E}^\mathbb{P}[(S_{1,t}-KS_{2,t})^+|\mathcal{F}_t]\\
&=S_{2,t}e^{(q_2-q_1)t}\mathbb{E}^\mathbb{Q}[e^{-q_1(T-t)}(R_{T}- Ke^{(q_1-q_2)T})^+|\mathcal{F}_t]\\
                                   &=S_{1,t}e^{-q_1(T-t)}\Big(\frac{1}{2}+\frac{1}{\pi}\int_0^{\infty} Re\Big[ \frac{f(u-i,X_t,T-t)}{f(-i,X_t,T-t)iu}\Big]du\Big)\nonumber \\
&- KS_{2,t}e^{-q_2(T-t)}\Big(\frac{1}{2}+\frac{1}{\pi}\int_0^{\infty} Re\Big[ \frac{f(u,X_t,T-t)}{iu}\Big]du\Big)\nonumber
\end{align}
\end{footnotesize}
\begin{remarque}
In case where the underlying are under a jump diffusion dynamic,
\begin{align*}
&dS_{i,t}=S_{i,t-} dX_{i,t},\\
&dX_{i,t}=(r-q_i)dt+\sigma dW_{i,t} + dZ_t^i\\
&Z^i_t=\int_0^t \int_{\mathbb{R}^2} e^{y_i}-1 J(ds,dy)\\
&d<W^1,W^2>_t=\rho dt.
\end{align*}
The dynamic of $X_t$ is:
\begin{align*}
dX_t&=-\Big( \frac{1}{2}\sigma^2-\int_{|y|<1} (y_1-y_2) \nu(ds,dy)\Big) dt\\
&+\sigma dW_t^\mathbb{Q}+\int_{|y|<1} (y_1-y_2) \tilde{J}(ds,dy)+\int_{|y|>1} (y_1-y_2) J(ds,dy),
\end{align*}
and the characteristic exponent under $\mathbb{Q}$
\begin{align*}
f(u,X_t,t)=\mathbb{E}[e^{iuX_t}]=exp\Big (- \frac{t}{2}\sigma^2  iu-\frac{1}{2}\sigma^2u^2t+ t\int_{\mathbb{R}^2} (e^{u(y_1-y_2)}-1) \nu(dy)\Big).
\end{align*}
\end{remarque}

\section{Characterization of the American option}
\noindent The natural price at time t of an American option denoted by $C(t,S_{1,t},S_{2,t})$ is writen as
\begin{align*}
C(S_{1,t},S_{2,t},t,T)&=S_{2,t}e^{(q_2-q_1)t}\sup_{\theta \in \mathcal{T}_{t,T}}\mathbb{E}^{\mathbb{Q}}[e^{-q_1(\tau-t) }(Ke^{(q_1-q_2)\tau}-R_\tau)^+|\mathcal{F}_t]\\
&=\sup_{\theta \in \mathcal{T}_{t,T}} \mathbb{E}^\mathbb{P}[e^{-r(\theta-t)}(KS_{1,t}-S_{2,t})^+|\mathcal{F}_t]\\
&=S_{2,t}e^{(q_2-q_1)t}u^A(t,R_t).
\end{align*}
We define $\tilde{u}^A(t,x) =  u^A(t,e^x)$ for all $(t,x)\in [0,T]\times \mathbb{R}$.
The Hamilton-Jacobi-Bellman (HJB in short) equation associated 
with $u^A(t,r)$ is a variational inequality involving, at least heuristically, a nonlinear second order parabolic integrodifferential equation (see
Bensoussan, J.L. Lions (1982) \cite{Bensoussan})
\begin{align}\label{eq:HJB}
\begin{cases}
min\{-\partial_t u^A-\mathcal{L^R}u^A+q_1u^A,u^A-(Ke^{(q_1-q_2)t}-r)^+\}=0 \text{ }\forall (t,r)\in [0,T)\times \mathbb{R}_+\\
u^A(T,r)=(Ke^{(q_1-q_2)T}-r)^+
\end{cases}
\end{align}
where
\begin{align*}
\mathcal{L}^R u(t,r)&=\frac{(\sigma r)^2}{2}\partial_{xx} u(t,r)-q_1 u(t,r)\\ 
&+\int_{\mathbb{R}^2}u(t,re^{y_1-y_2})-u(t,r)-r\partial_x u(t,r)\Big(e^{y_1-y_2}-1\Big)\tilde{\nu}(dt,dy) \big),
\end{align*}
or in the logarithmic representation
\begin{align}
\begin{cases}
min\{-\partial_t u-\mathcal{L^X}\tilde{u}^A+q_1\tilde{u}^A,\tilde{u}^A-(Ke^{(q_1-q_2)t}-e^x)^+\}=0 \text{ }\forall (t,x)\in [0,T)\times \mathbb{R}\\
\tilde{u}^A(T,r)=(Ke^{(q_1-q_2)T}-e^x)^+
\end{cases}
\end{align}
where
\begin{align*}
\mathcal{L}^Xu=\frac{\sigma^2}{2}(\partial_{xx} u(t,x)-\partial_{x} u(t,x)) +\int_{\mathbb{R}^2}u(t,x+y_1-y_2)-u(t,x)-\partial_x u(t,x)\big( e^{y_1-y_2}-1)\tilde{\nu}(dy) \big),
\end{align*}
The following classical lemma will be useful to study the continuity of $\partial_r u^A$ see \cite{Ladyzenskaja} and \cite{Jaillet}
\begin{lemma}\label{lemma1}
Let $u(t, x)$ be a function of $\mathbb{R}^2$
in $\mathbb{R}$, having partial derivatives $\partial_t u$ and $\partial_{xx} u$
uniformly bounded on $\mathbb{R}^2$. So, $\partial_x u$ verifies a Holder condition of order $\frac{1}{2}$ in $t$ uniformly with respect to x.
\end{lemma}
\begin{proposition}\label{prop4}
The function $\partial_r u^A$ (resp $\partial_x\tilde{u}^A$) is continuous in $[0,T]\times \mathbb{R}$ (resp. in $[0,T]\times \mathbb{R}_+)$ and
\begin{align*}
\lim_{(s,r)\rightarrow (t,b(t))} \partial_r u^A(s,r)=-1
\end{align*}
\end{proposition} 
\begin{proof}
we have
\begin{align*}
|\tilde{u}^A(t,x)-\tilde{u}^A(t,x')|&\leq C \mathbb{E}^\mathbb{Q}[sup_{\tau\in [t,T]}|X^{t,x}_{\tau}-X^{t,x'}_{\tau}|]\\
&\leq C  |x-x'|
\end{align*}
We now show continuity with respect to time for fixed $x$. Let $0 \leq t \leq t' \leq T$. Take $\tau \in \mathcal{T}_{T-t}$ and define $\tau'=\tau \wedge (T-t')$  . We note that $\tau' \in \mathcal{T}_{T-t'}$ and $\tau' \leq \tau \leq \tau\leq \tau' +t'-t$.
\begin{align*}
|u(t,x)-u(t',x)|&\leq C  sup_{\tau' \leq s\leq \tau' +t'-t} \mathbb{E}^\mathbb{Q}[|X^{t,x}_s-X^{t,x}_{\tau'}|].
\end{align*}
We have
\begin{footnotesize}
\begin{align*}
\mathbb{E}^\mathbb{Q}[|X_s^{x}-X_v^{x}|]&\leq \mathbb{E}^\mathbb{Q}[\Big|\Big( \int_{\mathbb{R}^2} e^{y_1-y_2}-1 \tilde{\nu}(dy)+\frac{1}{2}\sigma^2\Big)\Big| |s-v| ]+\mathbb{E}^\mathbb{Q}[\int_s^v \sigma dW_u]+\mathbb{E}^\mathbb{Q}[\int_s^t \int_{\mathbb{R}^2_+}(y_1-y_2)J(du,dy)]\\
&\leq C|s-v|+\mathbb{E}^\mathbb{Q}[Z_{1, s}-Z_{1, v}]-\mathbb{E}^\mathbb{Q}[Z_{2, s}-Z_{2, v}].
\end{align*}
\end{footnotesize} 
Where 
\begin{align*}
dZ_{i, s}=\int_{\mathbb{R}^2}y_i J(dy,ds)
\end{align*}
The condition \ref{condition} implies that the moments of $Z_{i,t}$ are
finite for all orders. Thus $Z_{i,t}$ is uniformly integrable. Since $Z_{i,t}$ is also continuous in probability, it is continuous in $L^1$ then $\mathbb{E}^\mathbb{Q}[Z_{i,s}-Z_{i,v}]\rightarrow 0$
then 
\begin{align*}
\mathbb{E}^\mathbb{Q}[|X_s^{x,v}-X_t^{x,v}|]&\leq C|s-v|
\end{align*}
and
\begin{align*}
|u(t,x)-u(t',x)|&\leq C|t'-t|
\end{align*}
we deduce that $\partial_x \tilde{u}^A$ is locally bounded in $[0, T] \times \mathbb{R}$
and $\partial_t \tilde{u}^{A}$ is  locally bounded in $(0, T] \times \mathbb{R}$. Which also gives  that  $\partial_r u^A$ is locally bounded in $[0, T] \times \mathbb{R}_+$
and $\partial_t  u^A$ is  locally bounded in $(0, T] \times \mathbb{R}_+$\\
Using  \ref{eq:HJB} we have$(\partial_t+ \mathcal{L}^R-q_1)u\leq 0.$ Then
\begin{align*}
 \frac{\sigma^2r^2}{2}\partial_{rr} u^A(t,r)&\leq q_1u^A(t,r)-\partial_t u^A(t,r)  \\
&-\int_{\mathbb{R}^2}u^A(t,re^{y_1-y_2})-u^A(t,r)-r\partial_r u^A(t,r)( e^{y_1-y_2}-1)\nu(dy) \big).
\end{align*}
We also have that 
\begin{align*}
 \big|\int_{\mathbb{R}^2}u^A(t,re^{y_1-y_2})-u^A(t,r)-r\partial_r u^A(t,r)( e^{y_1-y_2}-1)\tilde{\nu}(dy) \big) \leq C|r|(1+\partial_x \tilde{u}^A(t,x))
\end{align*}
where $C= \int_{\mathbb{R}^2} |e^{y_1-y_2}-1| \nu(dy)$, then 
$\partial_{xx}\tilde u^A(t,x)$ is locally bounded. Now we can apply the lemma \ref{lemma1} and we prove that $\partial_x \tilde{u}^A(t,x)$ is continuous in $[0,T]\times \mathbb{R}$ (resp $\partial_r u^A(t,r)$ is continuous in $[0,T]\times \mathbb{R}_+$
\end{proof}
\noindent A classical result shows that the domain  $[0,T] \times \mathbb{R}_{+}$ of the American put option price $u$ is therefore divided by the optimal-stopping boundary $\{(t, b(t)), t \in [0, T]\}$ into:
\begin{itemize}
\item The continuation region:
\begin{align*}
\mathcal{C}^p&=\{ (t,r) \in (0, T]\times \mathbb{R}_{+}, u(t,r)>Ke^{(q_1-q_2)t}-r \}\\
&=\{ (t,r) \in (0, T]\times \mathbb{R}_{+}, r>b(t) \}
\end{align*}
\item The exercise(Stopping) region:
\begin{align*}
\mathcal{S}^p&=\{ (t,r) \in (0, T]\times \mathbb{R}_{+}, u(t,r)=Ke^{(q_1-q_2)t}-r \}\\
&=\{ (t,r) \in (0, T]\times \mathbb{R}_{+}, r\leq b(t) \}.
\end{align*}
\end{itemize} 
\noindent and the variational inequality \ref{eq:HJB} can be written as
\begin{align*}
\begin{cases}
\partial_t u^A(t,r)+\mathcal{L}^Ru^A(t,r)-q_1u^A(t,r) =0 \hspace{1cm} \forall [0,T)\times (b(t),\infty)\\
u^A(t,r)=(Ke^{(q_1-q_2)t}-r)\hspace{1cm} \forall [0,T)\times (0,b(t))\\
u^A(T,r)=(Ke^{(q_1-q_2)T}-r)^+
\end{cases} 
\end{align*}
\begin{remarque}
Equation \ref{eq:HJB} gives us that $\partial_t u^A(t,r)+\mathcal{L}^Ru^A(t,r)-q_1u^A(t,r) \leq 0$   \hspace{0.4cm} $\forall [0,T)\times (0,b(t))$ 
\end{remarque}
\noindent In the next proposition, we show that the price of the American
exchange option, can be decomposed into a sum of the European
exchange option price and the early exercise premium. 
\begin{proposition}
The value of the American Exchange options has the representation
\begin{footnotesize}
 \begin{align*}
C (S_{1,t}, S_{2,t}, t, T)&=c (S_{1,t}, S_{2,t}, t, T)-q_1S_{2,t}e^{(q_2-q_1)t}\int_t^T\mathbb{E}^{\mathbb{Q}} \Big[ 1_{R_s\leq b(s)}(Ke^{(q_1-q_2)s}-R_{s-})|\mathcal{F}_t\Big]ds\\
&+S_{2,t}e^{(q_2-q_1)t}\int_t^T\mathbb{E}[\int_{b(s)e^{y_2-y_1}<R_{s-}<b(s)} u^A(s,R_{s-}e^{y_1-y_2})+Ke^{(q_1-q_2)s}-R_{s-}e^{y_1-y_2}
\tilde{\nu} (dy)ds|\mathcal{F}_t]ds.
\end{align*}
\end{footnotesize}
\end{proposition}
\begin{proof}
Knowing that $dR_t=R_{t-}\Big ( \sigma dW_t^\mathbb{Q}+ \int_{\mathbb{R}^2} (e^{y_1-y_2}-1) (J-\tilde{\nu})(dt,dy) \Big)$ and applying generalized Ito’s lemma (because the function $u^A$ is $\mathcal{C}^1$, piecewise $\mathcal{C}^2$ in $x$, and piecewise $\mathcal{C}^1$ in $t$ see proposition \ref{prop4}). we obtain
\begin{align*}
u^A(T,R_T)&=u^A(t,R_t)+\int_t^T \Big(\partial_{s} +\mathcal{L}^R\Big)u^A(t,R_{s-})ds
+\int_t^T \partial_{x} u(t,R_{s-})dR_s
\end{align*}
Separating the American option value into the regions, $u^A(t,R_t)=1_{R_t>b_t} u^A(t,R_t)+1_{R_t\leq b_t} (Ke^{(q_1-q_2)t}-R_t)$ we have 
\begin{footnotesize}
\begin{align*}
u^A(T,R_T)=e^{-q_1(T-t)}(R_T-Ke^{(q_1-q_2)T})^+&=u^A(t,R_t)+\int_t^T \Big(\partial_{s} +\mathcal{L}^R\Big) \big(1_{R_s>b(s)}u^A(s,R_s)+1_{R_s\leq b(s)}(R_s-Ke^{(q_1-q_2)s})\big)\\
&+\int_t^T\partial_R\big(1_{R_s>b(s)}u^A(s,R_s)+1_{R_s\leq b(s)}(R_s-Ke^{(q_1-q_2)s})\big)dR_s.
\end{align*}
\end{footnotesize}
On the continuation region, $u^A(t,R_t)$ satisfies $\Big(\partial_t+\mathcal{L}^R-q_1\Big) u^A(t,R_t)=0$. Taking expectation with respect to the measure $\mathbb{\mathbb{Q}}$ and using the fact that $R_t$ is a martingale we have 
\begin{footnotesize}
\begin{align*}
u^E(t,R_t)&=u^A(t,R_t)-q_1\int_t^T\mathbb{E}^{\mathbb{Q}} \Big[ 1_{R_s\leq b(s)}(Ke^{(q_1-q_2)s}-R_s)|\mathcal{F}_t\Big]ds\\
&+\int_t^T\mathbb{E}[\int_{b(s)e^{y_2-y_1}<R_{s-}<b(s)} u^A(s,R_{s-}e^{y_1-y_2})+Ke^{(q_1-q_2)s}-R_{s-}e^{y_1-y_2}
\tilde{\nu} (dy)ds|\mathcal{F}_t]ds.
\end{align*}
\end{footnotesize}
This conclude the proof of the proposition.
\end{proof}
\begin{remarque}
In case if the underlyings are driven by a jump diffusion, the american option can be represented as following:
\begin{footnotesize}
 \begin{align*}
C (S_{1,t}, S_{2,t}, t, T)&=c (S_{1,t}, S_{2,t}, t, T)-q_1S_{2,t}e^{(q_2-q_1)t}\int_t^T\mathbb{E}^{\mathbb{Q}} \Big[ 1_{R_s\leq b(s)}(Ke^{(q_1-q_2)s}-R_{s-})|\mathcal{F}_t\Big]ds\\
&+S_{2,t}e^{(q_2-q_1)t}\int_t^T\mathbb{E}[\int_{b_se^{y_2-y_1}<R_{s-}<b_s} u^A(s,R_{s-}e^{y_1-y_2})+Ke^{(q_1-q_2)s}-R_{s-}
\nu (dy)ds|\mathcal{F}_t]ds.\\
&-S_{2,t}e^{(q_2-q_1)t}\int_t^T\mathbb{E}[\int_{R_{s-}<b_s} R_{s-}(e^{y_1-y_2}-1)
\nu (dy)ds|\mathcal{F}_t]ds.
\end{align*}
\end{footnotesize}
\end{remarque}
\subsection{Properties of the free boundary}
Throughout this section we will prove some properties of the free boundary. Notice that since $t\mapsto u^A(t,.)$ is non increasing the function $t\mapsto b(t)$ is non-decreasing. 
\begin{proposition}
 For $t \in [0,T)$, we have $b(t) > 0$.
\end{proposition}
\begin{proof}
Suppose that $b(t^*) = 0$ for some
$t^* \in (0,T)$, we then have $b(t) = 0$ for $t \leq t^*$, and 
\begin{align*}
\tilde{u}^A(t,x)>(Ke^{(q_1-q_2)t}-e^x) \hspace{0.3cm}\text{and} \hspace{0.3cm}(\partial_t+\mathcal{L}^X-q_1)\tilde{u}^A=0 \hspace{0.3cm} \forall t\in (0,t^*)\times \mathbb{R}.
\end{align*}
We know  that $\tilde{u}^A(.,x)$ is non-increasing, then $(\mathcal{L}^X-q_1)\tilde{u}^A\geq 0.$ Let $\theta \in \mathcal{C}_{c}^{\infty}(0,t^*)$ and $\phi \in \mathcal{C}_{c}^{\infty}(\mathbb{R})$, then we have
\begin{align*}
\int_{(0,t^*)}\theta(t)\int_{\mathbb{R}}u(t,x)(\sigma^2(\partial_{xx}\phi(x)-\partial_{x}\phi(x))+ B^*(\phi)dx dt&\geq q_1\int_{(0,t^*)}\theta(t)\int_{\mathbb{R}}u(t,x)\phi (x) dx dt\\
&\geq q_1\int_{(0,t^*)}\theta(t)\int_{\mathbb{R}}(Ke^{(q_2-q_1)t}-e^x)\phi(x) dx dt
\end{align*}
where
\begin{align*}
B^*(\phi)=\int_{\mathbb{R}^2}\phi(x+y_1-y_2)-\phi(x)+\partial_x \phi(x)\big( e^{y_1-y_2}-1)\nu(dy)
\end{align*}
Let $\chi \in \mathcal{C}^{\infty}$ such that $supp(\chi) = [-1,0]$ and $\int \chi(x)dx=1$.  By setting $\phi(x)=\lambda \chi(\lambda x )$ then $q_1\int_{\mathbb{R}}(Ke^{(q_2-a_1)t}-e^x)\phi(x) dx =q_1K e^{(q_2-q_1)t}-\int_{\mathbb{R}}e^{x/ \lambda }\chi(x)dx$. Letting $\lambda \rightarrow 0$ we had
 \begin{align*}
q_1\int_{\mathbb{R}}(Ke^{(q_2-q_1)t}-e^x)\phi(x) dx \xrightarrow[\lambda \rightarrow 0] {}q_1K e^{(q_2-q_1)t},
\end{align*}
since $supp(\chi) = [-1,0]$, we had $\lim_{\lambda \rightarrow 0} \int_{\mathbb{R}}e^{x/ \lambda }\chi(x)dx=0$.
As $\tilde{u}^A(t,x)$ is bounded,
\begin{align*}
\int_{\mathbb{R}}\tilde{u}^A(t,x)(\sigma^2(\partial_{xx}\phi(x)-\partial_{x}\phi(x)) dx
&\leq \|u\|_{\infty}\int_{\mathbb{R}}\sigma^2(\lambda^2 \chi''(x)-\lambda \chi'(x)) dx\\
&\xrightarrow[\lambda\rightarrow 0]{}0.
\end{align*}
We also have
\begin{footnotesize}
\begin{align*}
\int_{\mathbb{R}}\tilde{u}^A(t,x)B^*(\phi)(x) dx \leq \|u^A\|_{\infty}\int_{\mathbb{R}^3} \chi(x+y_1-y_2)-\chi(x)+\lambda (e^{y_1-y_2}-1)\chi'(x)  \tilde{\nu}(dy)dx
\end{align*}
\end{footnotesize}
As $\chi \in \mathcal{C}^{\infty}$ then $|\chi(x+y_1-y_2)-\chi(x)+\lambda (e^{y_1-y_2}-1)\chi'(x)\big|\leq   2| \chi' \|_{\infty}\int_{\mathbb{R}^2}|y_1-y_2| +\lambda(e^{y_1-y_2}-1)\tilde{\nu}(dy)$ then by dominated convergence we have 
\begin{align*}
\int_{\mathbb{R}^2}\tilde{u}^A(t,x)B^*(\phi)(x) dx \xrightarrow[\lambda \rightarrow 0]{} 0
\end{align*}
We conclude by sending $\lambda \rightarrow 0$ we had $0> q_1$ which is a condraction 
\end{proof}
\begin{proposition}

For $t \in [0,T)$, we have $b(t)$ is continuous in $[0,T]$.

\end{proposition}

\begin{proof}

Since $t\rightarrow u^A(t,r)$ is nonincreasing, the function $t\rightarrow b(t)$ is nondecreasing. Let $t\in [0,T]$, we construct a decreasing sequence $(t_n)_{n\geq 1}$ such that $\lim_{n\rightarrow 0} t_n=t$. As $b(t)$ is nondecreasing then we have $\lim_{n\rightarrow 0} b(t_n) \geq b(t)$. We know that in the exercice region we have

\begin{align*}
u^A(t_n,b(t_n))=(Ke^{(q_1-q_2)t_n}-b(t_n)),
\end{align*}
as the functions above are continuous, then
\begin{align*}
u^A(t,\lim_{n\rightarrow 0}b(t_n))=(Ke^{(q_1-q_2)t}-\lim_{n\rightarrow 0} b(t_n)),
\end{align*}
then $\lim_{n\rightarrow 0} b(t_n) \leq b(t)$. So we proved the right-continuity.\\
We now prove that $b$ is left-continuous. let $\bar{b}(t)=\log b(t)$, we also have that $\bar{b}$ is nondecreasing and $\bar{b}(t^-)\leq \bar{b}(t)$. Suppose that $\bar{b}(t^-)< \bar{b}(t)$.\\
let $(s,x)\in (0,t)\times (\bar{b}(t^-),\bar{b}(t))$ then $x\geq \bar{b}(s)$ so
\begin{align*}
(\partial_t+\mathcal{L}^X -q_1)\tilde{u}^A(s,x)=0 \hspace{1cm} \forall (s,x)\in [0,t)\times (\bar{b}(t^-),\bar{b}(t))
\end{align*}
We know that $\tilde{u}^A$ is decreasing in time we deduce that for every $s \in(0,t)$:
\begin{align*}
(\mathcal{L} -q_1)\tilde{u}^A(s,x)\geq 0 \hspace{1cm} x\in (\bar{b}(t^-),\bar{b}(t))
\end{align*}
As $u^A$ is continuous in time then
\begin{align}\label{eq:1}
(\mathcal{L} -q_1)\tilde{u}^A(t,x)\geq 0 \hspace{1cm} x\in (\bar{b}(t^-),\bar{b}(t))
\end{align}
As $ x \in (\bar{b}(t^-),\bar{b}(t))$ we have $\tilde{u}^A(t,x)=Ke^{(q_1-q_2)t}-e^{x}$ then by substituing in \ref{eq:1} $\tilde{u}^A(t,x)$ by $Ke^{(q_1-q_2)t_n}-e^{x}$ we had for all $x \in (\bar{b}(t^-),\bar{b}(t))$
\begin{align*}
\int_{\mathbb{R}^2}\tilde{u}^A(t,x+y_1-y_2)-Ke^{(q_1-q_2)t}- e^{x+y_1-y_2}\nu(dy)
-q_1(Ke^{(q_1-q_2)t}-e^{x})\geq 0. \hspace{1cm}
\end{align*}
Let $(s,x)\in (t,T)\times (\infty,\bar{b}(t))$ then $\tilde{u}^A(s,x)=Ke^{(q_1-q_2)s}-e^{x}$ as we also know that $(\partial_s +\mathcal{L}-q_1)\tilde{u}^A(s,x)\leq 0$ then
\begin{align*}
(\mathcal{L}-q_1)\tilde{u}^A(s,x)\leq 0,
\end{align*}
or for all $(s,x)\in (t,T)\times (\infty,\bar{b}(t))$
\begin{align*}
\int_{\mathbb{R}^2}\tilde{u}^A(s,x+y_1-y_2)-Ke^{(q_1-q_2)s}+ e^{x+y_1-y_2}\nu(dy)
-q_1(Ke^{(q_1-q_2)s}-e^{x})\leq 0.\hspace{1cm}
\end{align*}
As ${u}^A$ is continuous, letting $s\rightarrow t$ then $\int_{\mathbb{R}^2}\tilde{u}^A(t,x+y_1-y_2)-Ke^{(q_1-q_2)t}- e^{x+y_1-y_2}\nu(dy)
-q_1(Ke^{(q_1-q_2)t}-e^{x})\leq 0$ for all $x\in (\infty,\bar{b}(t))$.\\ We can now conclude that for all $x\in (\bar{b}(t^-),\bar{b}(t))$
\begin{align*}
\int_{\mathbb{R}^2}\tilde{u}^A(t,x+y_1-y_2)-Ke^{(q_1-q_2)t}- e^{x+y_1-y_2}\nu(dy)
-q_1(Ke^{(q_1-q_2)t}-e^{x})= 0.\hspace{1cm}
\end{align*}
Setting $r=e^x$ then 
\begin{align}\label{eq:4}
\int_{\mathbb{R}^2}u^A(t,re^{y_1-y_2})-Ke^{(q_1-q_2)t}+ re^{y_1-y_2}\nu(dy)
-q_1(Ke^{(q_1-q_2)t}-r)= 0,\hspace{0.5cm} \forall x\in (b(t^-),b(t)).
\end{align}
We define $f(r)=\int_{\mathbb{R}^2}u^A(t,re^{y_1-y_2})+re^{y_1-y_2}\nu(dy)+q_1r= 0$ where $r\in [0,b(t))$. As $u^A(t,.)$ is continous and convex then $f$ is also continuous, convex and non negative. We had that $f(0)=0$ and $f>0$ in $ (b(t^-),b(t))$, then $f$ must be strictly increasing which contradicts \ref{eq:4} as we have supposed $b(t^-) < b(t)$.
\end{proof}
\noindent The following result characterizes the limit of the critical price $b(t)$ as t approaches $T$.
\begin{proposition}
we have $\lim_{t\rightarrow T} b(t) = S_0$
where $S_0$ the unique real number in the interval $(0,Ke^{(q_1-q_2)T})$ such that
\begin{align*}
\int_{\mathbb{R}^2}(Ke^{(q_1-q_2)T}-S_0e^{y_1-y_2})^+\nu(dy)=q_1(Ke^{(q_1-q_2)T}-S_0)
\end{align*}
\end{proposition}
\begin{proof}Let $t\in (0,T)$. Define $b(T ) = \lim_{t\rightarrow T} b(t)$ and $\bar{b}(T ) = \log b(T )$. We clearly have $b(T ) \leq Ke^{(q_1-q_2)T}$.
Let $x>log (b(t))$ then we have $(\partial_t +\mathcal{L}^X-q_1)\tilde{u}^A(t,x)=0$ as $u$ is nonincreasing in time then $(\mathcal{L}^X-q_1)\tilde{u}^A(t,x)\geq 0$. Note that $(\mathcal{L}^X-q_1)\tilde{u}^A(t,x)\xrightarrow[t\rightarrow T]{} (\mathcal{L}^X-q_1)(Ke^{(q_1-q_2)T}-e^x)^+ $ (in distribution sens) then for all $x>logb(T)$ 
\begin{align*}
\mathcal (\mathcal{L}^X-q_1) (Ke^{(q_1-q_2)T}-e^x) \geq 0,
\end{align*}
so when $x \in (-\infty , log(K)+(q_1-q_2)T)$ we have
\begin{align*}
\int_{\mathbb{R}^2}(Ke^{(q_1-q_2)T}-e^{x+y_1-y_2})^+-Ke^{(q_1-q_2)T}+ e^{x+y_1-y_2})\nu(dy)-q_1(Ke^{(q_1-q_2)T}-e^x)\geq 0.
\end{align*}
then 
\begin{align*}
\int_{\mathbb{R}^2}(e^{x+y_1-y_2}-Ke^{(q_1-q_2)T})^+-q_1(Ke^{(q_1-q_2)t}-e^x\geq 0.
\end{align*}
Otherwise if $x<log (b(t))$, we know that  $\partial_t \tilde{u}^A+\mathcal{L}^X\tilde{u}^A-q_1\tilde{u}^A\leq 0$ and $\tilde{u}^A(t,x)=Ke^{(q_1-q_2)t}-e^x$
\begin{align*}
\int_{\mathbb{R}^2}u(t,x+y_1-y_2)+e^{x+y_1-y_2}-Ke^{(q_1-q_2)T}\nu(dy)-q_1(Ke^{(q_1-q_2)t}-e^x)\leq 0. 
\end{align*}
By letting $t\rightarrow T$ the for $x<log (b(T))$ we had 
\begin{align*}
\int_{\mathbb{R}^2}(Ke^{(q_1-q_2)T}-e^{x+y_1-y_2})^++e^x\big( e^{y_1-y_2})-Ke^{(q_1-q_2)T}\nu(dy)-q_1Ke^{(q_1-q_2)t}-e^x)\leq 0,
\end{align*}
thus
\begin{align*}
\int_{\mathbb{R}^2}(Ke^{(q_1-q_2)T}-e^{x+y_1-y_2})^+\nu(dy)-q_1Ke^{(q_1-q_2)t}-e^x)\leq 0. 
\end{align*}
Then we conclude that if $x=\log b(T)$
then
\begin{align*}
\int_{\mathbb{R}^2}(Ke^{(q_1-q_2)T}-b(T)e^{y_1-y_2})^+\nu(dy)=q_1(Ke^{(q_1-q_2)T}-b(T)).
\end{align*}

\end{proof}

\subsubsection{An approximate formula for an
American exchange option}
As an analytic solution for a European exchange option is known in \ref{eq:5}, the remaining problem is to derive a good approximation for the early exercise premium.
\begin{proposition}
There exists an $\alpha< 0$  such that
\begin{align*}C(S_{1,t},S_{2,t},t,T)=
\begin{cases}
&c(S_{1,t},S_{2,t},t,T)+S_{2.t}e^{(q_2-q_1)t}h(T-t)A(h)r^\alpha \text{ if  }r>b(t),\\
&S_{2.t}e^{(q_2-q_1)t}(Ke^{(q_1-q_2)t}-r) \text{ if  } r\leq b(t)
\end{cases}
\end{align*}
Where 
\begin{align*}
r=e^{(q_1-q_2)t} \frac{S_{1.t}}{S_{2.t}}
\end{align*}
$h(t)=1-e^t$, and $A(h)$ with the early‐exercise boundary b(t) satisfies:
\begin{align*}
&A(h)=-\frac{1+\partial_r u^E(t,b(t))}{h(T-t)\alpha b(t)^{\alpha-1}}\\
&u^E(t,b(t))-\frac{b(t)}{\alpha h(T-t)}(1-\partial_r u^E(t,b(t)))=Ke^{(q_1-q_2)t}-b(t) .
\end{align*}
\end{proposition}
\begin{proof}
Following proposition 5 we know that
\begin{align}\label{eq:6}
u^A(t,x)=u^E(t,x)+u^P(t,x)
\end{align}
the remaining problem is to derive a good approximation for the early exercise premium.  Given the linearity of \ref{eq:6}, the
early‐exercise premium, $u^P$, must satisfy the Equation: 
 \begin{align}\label{eq:up}
\begin{cases}
\partial_t u^P(t,r) dt+\frac{\sigma^2r^2}{2} \partial_{rr}u-q_1u^P(t,r)\\
+\int_{\mathbb{R}^2}u(t,re^{y_1-y_2})-u(t,r)-\partial_r u(t,r)(e^{ y_1-y_2}-1)\tilde{\nu}(dy)\text{ for } r>  b(t)\\
u^p(t,r)=Ke^{(q_1-q_2)t}-r-u^E(t,x)\text{ for } r\leq b(t)\\
u^p(t,r) \text{ is continuous }\\
\frac{du^p}{dr} \text{ is continuous }\\
\lim_{r\to \infty} u^p(t,r)=0
\end{cases}
\end{align}
Following the method of MacMillan \cite{MacMillan}, we replace the variable t by $h(t)=1-e^{-t}$ and we rewrite $u^P(t,x)$ in the form: $u^P(t,x)=h(T-t)g(h(T-t),x)$, Substituting this
structural form into \ref{eq:up}, we have that
\begin{align*}
(1-h(T-t))\big(\frac{g}{h}-\frac{\partial g}{\partial h}\big)=(\mathcal{L}^R-q_1)g,
\end{align*}
As described in Barone‐Adesi and Whaley \cite{Barone}, $(1-h(T-t))\frac{g}{h}-\frac{\partial g}{\partial h}$ is negligible since $1- h(T- t)$ approaches 0 when $T- t$ is small. Consequently, we ignore the term $(1-h(T-t))\frac{g}{h}-\frac{\partial g}{\partial h}$\\
The approximate equation for $g$ now becomes 
\begin{align*}
\frac{h'(T-t)}{h(T-t)}g=\big(\mathcal{L}^R-q_1\big)g,
\end{align*}
or in logarithmic representation
\begin{align}\label{eq:2}
\frac{h'(T-t)}{h(T-t)}g=\big(\mathcal{L}^X-q_1\big)g.
\end{align}
We remark that the time t is not an explicit variable in right side of Equation \ref{eq:2}. Instead, t appears only in $\frac{h'(T-t)}{h(T-t)}$which is the coefficient of term $g$ in the left‐hand side. This implies that $t$, or equivalently $T- t$, is
treated as a parameter. We assume the form of the solution for $g$ in the continuation region is
\begin{align*}
g(h,r)=A(h)r^{\alpha}.
\end{align*}
 In logarithmic representation we had the following form:
\begin{align*}
g(h,x)=A(h)e^{\alpha x}.
\end{align*}
Moreover, we need that $\alpha$ to be negative because $u^P(t,r)\to 0$ as $r\to \infty$ . Substituing $g$ in \ref{eq:2} by $A(h)e^{\alpha x}$ we have
\begin{align*}
\frac{\sigma^2}{2}\alpha^2 -\Big(\frac{\sigma^2}{2}+\int_{\mathbb{R}^2}e^{y_1-y_2}-1 \text{ }
 \tilde{\nu}(dy)\Big)\alpha+\int_{\mathbb{R}^2}e^{\alpha (y_1-y_2)} -1\text{ }\tilde{\nu}(dy) -q_1=\frac{h'(T-t)}{h(T-t)}.
\end{align*}
Let 
\begin{align*}
f(\alpha)=\frac{\sigma^2}{2}\alpha^2 -\Big(\frac{\sigma^2}{2}+\int_{\mathbb{R}^2}e^{y_1-y_2}-1 \text{ }
 \tilde{\nu}(dy)\Big)\alpha+\int_{\mathbb{R}^2}e^{\alpha (y_1-y_2)} -1\text{ }\tilde{\nu}(dy) -q_1-\frac{h'(T-t)}{h(T-t)}
\end{align*}
we have that
\begin{align*}
f'(\alpha)&=\sigma^2\alpha -(\frac{\sigma^2}{2}+\int_{\mathbb{R}^2}e^{(y_1-y_2)}\tilde{\nu}(dy)+1)  +\int_{\mathbb{R}^2}(y_1-y_2)e^{\alpha(y_1-y_2)}\tilde{\nu}(dy)\\
f''(\alpha)&=\sigma^2 +\int_{\mathbb{R}^2}(y_1-y_2)^2e^{\alpha(y_1-y_2)}\tilde{\nu}(dy)>0
\end{align*}
This shows that the function $f$ is convex. We have $f(0)=-q_1-\frac{h'(T-t)}{h(T-t)}<0$ so as $\lim_{\alpha \rightarrow \pm \infty }f(\alpha)=+\infty$ we can say that we have one unique negative solution to $f=0$.\\
Because $u^P(t,r)$ and $\partial_x u^P(t,r)$ are continuous
\begin{align*}
&u^E(t,b(t))+h(T-t)A(h) b(t)^{\alpha}=Ke^{(q_1-q_2)t}-b(t) \\
&\alpha h(T-t)A(h)b(t)^{\alpha-1}=-1-\partial_r u^E(t,b(t)).
\end{align*}
Thus
\begin{align*}
&A(h)=-\frac{1+\partial_r u^E(t,b(t))}{h(T-t)\alpha b(t)^{\alpha-1}}\\
&u^E(t,b(t))-\frac{b(t)}{\alpha h(T-t)}(1-\partial_r u^E(t,b(t)))=Ke^{(q_1-q_2)t}-b(t) .
\end{align*}
\end{proof}
\section{Summary and conclusion}
In this article, we first present a closed‐form solution for the value of a European exchange option in a jump‐diffusion
model and L\'evy model by reducing the dimension of the model by a change of measure. American exchange option price can also be represented as the sum of the
price of the corresponding European exchange option price
and an early exercise premium, similar to the findings of Cheang
and Chiarella \cite{Cheang} in Jump diffusion case, however Cheang and Chiarella did not show the regularity of American option to justify the use of Ito lemma. We were also able to show different properties of the free boundary thanks to dimension reduction. Finally we give an approximate formula of an American exchange option.
 
\end{document}